\documentclass[twocolumn,5p,sort&compress,times]{elsarticle}

\usepackage{amsmath,amssymb}
\usepackage[utf8]{inputenc}
\usepackage{graphicx}
\usepackage{bm}
\usepackage{hyperref}
\usepackage{dsfont}
\usepackage{braket}
\usepackage{xcolor}
\usepackage{mathtools}

\usepackage{amsthm}
\theoremstyle{definition}
\newtheorem{thm}{Theorem}
\newtheorem{lem}{Lemma}

\newtheorem{cor}{Corollary}

\newcommand{\complex}{\mathbb{C}}
\newcommand{\Hil}{\mathcal{H}} 
\newcommand{\Bnd}[1]{\mathcal{B} (#1)}

\newcommand{\Trone}[1]{\mathcal{T}_{1} (#1)}
\newcommand{\St}[1]{\mathcal{S} (#1)} 
\newcommand{\Eff}[1]{\mathcal{E} (#1)} 
\newcommand{\id}{\mathds{1}}
\newcommand{\tr}[1]{\operatorname{tr}\left[#1\right]} 
\newcommand{\ptr}[2]{\operatorname{tr}_{#1}\left[#2\right]} 
\newcommand{\obM}{\mathsf{M}}
\newcommand{\obA}{\mathsf{A}}

\begin{document}

\begin{frontmatter}
\title{Separability Criterion for Quantum Effects}

\author{Ikko Hamamura}
\ead{hamamura@nucleng.kyoto-u.ac.jp}
\address{Department of Nuclear Engineering, Kyoto University, 6158540 Kyoto, Japan}

\begin{abstract}
  Entanglement of quantum states is absolutely essential for modern quantum sciences and technologies.
  It is natural to extend the notion of entanglement to quantum observables dual to quantum states.
  For quantum states, various separability criteria have been proposed to determine whether
  a given state is entangled.
  In this Letter, we propose a separability criterion for specific quantum effects (binary
  observables) that can be regarded as a dual version of the Bell--Clauser--Horne--Shimony--Holt (Bell--CHSH) inequality
  for quantum states.
  The violation of the dual version of the Bell--CHSH inequality is confirmed by using IBM's cloud quantum computer.
  As a consequence, the violation of our inequality rules out the maximal tensor product state space,
  that satisfies information causality and local tomography.
  As an application, we show that an entangled observable which violates our inequality is useful for quantum teleportation.
\end{abstract}
\end{frontmatter}

\section{Introduction}

According to the axioms of quantum theory introduced by von Neumann~\cite{Neumann1932},
a quantum system is associated with a separable Hilbert space
and a composite quantum system is given by a tensor product of Hilbert spaces.
In recent years, research to derive the axiom of Hilbert space from the physical or informational principle has gained significant momentum~\cite{Chiribella2011,Dakic2011,Hardy2016,Masanes2011,Wilce2012}.
Once the axiom of the local Hilbert space is derived, observable algebra (or positive operator valued measures as observables) and Born’s probability rule can immediately be derived.
However, the composite system cannot be determined uniquely merely from the axiom of the local Hilbert space, so its physical justification is desired.
One of the keys to characterize the composite system is the entanglement of observables
because it is not compatible with the maximal tensor product state space
while it satisfies information causality~\cite{Pawlowski2009} and local tomography~\cite{Barrett2007}.

In this work, we develop a separability criterion for a certain class of observables. 
For quantum states, several separability criteria are already known, such as the positive partial transpose criterion~\cite{Peres1996, Horodecki1996}, range criterion~\cite{Horodecki1997}, reduction criterion~\cite{Horodecki1999,Cerf1999},
and entanglement witness~\cite{Horodecki1996} represented by the Bell--Clauser--Horne--Shimony--Holt (Bell--CHSH) witness~\cite{Clauser1969}.
Conversely, only a few attempts have been done for observables~\cite{Vertesi2011, Rabelo2011}.
Our main purpose is to establish the dual version of the Bell--CHSH inequality
so that one can detect entanglement of observables from violations of the inequality.
Such a violation is experimentally confirmed by using IBM's cloud quantum computer.
As an application, we show that observables that violate our inequality are useful for quantum teleportation~\cite{Bennett1993}.

The outline of this Letter is as follows:
We begin by introducing the three types of positivities and fixing notations in Sec.~\ref{sec:threepos}.
Using these positivities, we discuss the state spaces of the composite systems.
In this Letter, we assume that a local quantum system is given by a finite-dimensional Hilbert space $\Hil=\complex^d$.
In Sec.~\ref{sec:nosigloctom}, we recall that some physical principles for the composite systems determine a family of possible state spaces
that contains the minimal, maximal, and ``physical'' tensor products associated with the tensor product of Hilbert spaces $\complex^{d_A d_B}\cong\complex^{d_A}\otimes \complex^{d_B}$.
While the violation of the Bell--CHSH inequality rules out the minimal one,
it does not rule out the maximal one.
To overcome this problem, we propose in Sec.~\ref{sec:dualbell} a dual version of the Bell--CHSH inequality that can be used to exclude the maximal one.
This inequality can be violated by an \textit{entangled effect} whose definition is explained below.
In Sec.~\ref{sec:experiment}, we show experimentally by using a quantum computer (IBM Quantum Experience) that an entangled effect violating this new inequality exists.
Because the maximal composite system does not allow any entangled effect to exist,
we exclude the maximal composite system from the possible candidates.
In Sec.~\ref{sec:qtelepo}, we show that this violation is useful for quantum teleportation.
Finally, some concluding comments are given in Sec.~\ref{sec:conclusion}.

\section{Three types of positivities}\label{sec:threepos}
We introduce three positivities that are convenient for discussing the composite systems.
An operator $X$ on a Hilbert space $\Hil$ is \textit{positive},
denoted by $X \geq 0$, if $\Braket{\psi|X \psi}\geq 0$ holds for all $\ket{\psi} \in \Hil$.
A bipartite operator $X$ on a Hilbert space $\Hil_A\otimes\Hil_B$ is \textit{positive on pure tensors (POPT)}~\cite{Barnum2005},
denoted by $X \geq_{\mathrm{POPT}} 0$,
if $\Braket{\psi \otimes \phi| X \psi \otimes \phi} \geq 0$ holds for all $\ket{\psi}\in \Hil_A$ and $\ket{\phi} \in \Hil_B$.
POPT is also called block positive~\cite{Labuschagne2006,Szarek2008}.
A non-positive POPT operator is called an ``entanglement witness''~\cite{Horodecki1996}.
A bipartite operator $X$ on a Hilbert space $\Hil_A\otimes\Hil_B$ is \textit{separable positive} and is denoted by $X \geq_{\mathrm{SP}} 0$ if $X$ has a decomposition $X = \sum_i A_i \otimes B_i$ such that each $A_i$ and $B_i$ is a positive operator on $\Hil_A$ and $\Hil_B$ respectively.
Since $ \sum_i A_i \otimes B_i$ is a positive operator for positive operators $A_i$ and $B_i$,
$X \geq_{\mathrm{SP}} 0$ implies $X \geq 0$.
Since the positive operator is positive on pure tensors,
$X \geq 0$ implies $X \geq_{\mathrm{POPT}}0$.
However, the converses do not hold.

On the set of operators on a finite-dimensional Hilbert space,
we introduce the Hilbert--Schmidt inner product as $\Braket{A|B}_{\mathrm{HS}}=\tr{A^* B}$.
The set of all positive operators is dual to itself via this inner product and
the set of all POPT operators and the set of all separable positive operators are dual to each other.

\section{No-signaling principle and local tomography}\label{sec:nosigloctom}
We assume that a local quantum system is associated with a finite-dimensional Hilbert space $\Hil = \complex^d$.
If an operator $\rho$ is positive and its trace is equal to unity,
the operator $\rho$ is called a state, which represents an experimental situation.
A state space $\St{\Hil}$ is the set of all states $\Set{\rho\in \Trone{\Hil}|\rho\geq 0}$,
where $\Trone{\Hil}$ is the set of trace class operators with unit trace.
An effect space is a dual space to the state space $\St{\Hil}$
and is defined as $\Eff{\Hil}=\Set{X\in\Bnd{\Hil}|0\leq X \leq \id}$.
An effect represents a measurement apparatus that outputs ``yes'' or ``no''.
Therefore an effect can be identified with a binary observable.
When the state $\rho$ is prepared,
the occurrence probability of a measurement event represented by $E \in \Eff{\Hil}$ is given by the trace formula $\tr{\rho E}$.
This trace formula is called the generalized Born rule.
From the no-signaling principle and local tomography (the latter was introduced as the global state assumption in Ref.~\cite{Barrett2007}),
the state space of the two quantum systems is bounded~\cite{Barnum2005,Barrett2007,Barnum2012,Janotta2014} by the minimal tensor product state space and the maximal tensor product state space, which are defined as follows:
A minimal tensor product space is
\begin{equation}
  \St{\complex^{d_A}}\otimes_{\min} \St{\complex^{d_B}}=\Set{\rho
     \in \Trone{\complex^{d_A d_B}}
    | \rho\geq_{\mathrm{SP}} 0},
\end{equation}
and
a maximal tensor product space is
\begin{equation}
  \St{\complex^{d_A}}\otimes_{\max} \St{\complex^{d_B}}=\Set{\rho
   \in \Trone{\complex^{d_A d_B}}
   | \rho\geq_{\mathrm{POPT}} 0}.
\end{equation}
Note that the minimal tensor product space corresponds to the set of separable states.
A state that is not a separable state is called an entangled state.
An element of the maximal tensor product state space is called a POPT state~\cite{Barnum2005,Barnum2010}.
The state space given by the axiom of tensor product $\St{\complex^{d_A}\otimes \complex^{d_B}}$ coincides with neither of them.
That is, $\St{\complex^{d_A}}\otimes_{\min} \St{\complex^{d_B}} \subsetneq \St{\complex^{d_A}\otimes \complex^{d_B}} \subsetneq \St{\complex^{d_A}}\otimes_{\max} \St{\complex^{d_B}}$ holds.
Note also that this problem does not occur in classical theory.
The tensor product of classical state spaces is unique~\cite{Namioka1969}.

The problem here is how we can distinguish $\St{\complex^{d_A}\otimes \complex^{d_B}}$ experimentally from other state spaces.
The left inequality is attained by the Bell's argument~\cite{Bell1964}.
Let $B$ be
\begin{equation}
  B = \tr{\rho (A_0B_0 + A_0 B_1+ A_1 B_0 - A_1 B_1)},
\end{equation}
then the Bell--CHSH inequality~\cite{Clauser1969} means $|B| \leq 2$ for any separable state $\rho$ and operators $A_i$ and $B_i$ satisfying $A_i^2=B_i^2=\id$ and $[A_i, B_j]=0$.
Conversely, there exists an entangled state $\rho$ and
pairs of incompatible measurements $\Set{A_0, A_1}$ and $\Set{B_0, B_1}$
that lead to the violation of the Bell--CHSH inequality~\cite{Wolf2009}.
The violation of the Bell--CHSH inequality thus implies the existence of an entangled state;
namely, the state space of the composite system is strictly larger than $\St{\complex^{d_A}}\otimes_{\min}\St{\complex^{d_B}}$.

\section{Dual Bell--CHSH inequality}\label{sec:dualbell}
It is not a straightforward task to distinguish $\St{\complex^{d_A}\otimes\complex^{d_B}}$ and $\St{\complex^{d_A}}\otimes_{\max} \St{\complex^{d_B}}$
because there exists a quantum mechanical representation for POPT states~\cite{Barnum2010}.
The maximum value of the left-hand side (LHS) of the Bell--CHSH inequality in the maximal tensor product state space is $2\sqrt{2}$, which coincides with that in $\St{\complex^{d_A}\otimes\complex^{d_B}}$.
Furthermore, the information causality~\cite{Pawlowski2009} is not strong enough to discard the maximal tensor product state space $\St{\complex^{d_A}}\otimes_{\max}\St{\complex^{d_B}}$.
To overcome this problem, we use the duality between a state space and an effect space.
We distinguish $\St{\complex^{d_A}\otimes\complex^{d_B}}$ and $\St{\complex^{d_A}}\otimes_{\max}\St{\complex^{d_B}}$ by considering their effect spaces.
The effect space dual to $\St{\complex^{d_A}}\otimes_{\max}\St{\complex^{d_B}}$ is the set of \textit{separable effects} $\Eff{\complex^{d_A}}\otimes_{\min}\Eff{\complex^{d_B}} = \Set{X \in \Bnd{\complex^{d_A}\otimes \complex^{d_B}}|0\leq_{\mathrm{SP}}X\leq_{\mathrm{SP}}\id}$.
An effect that is not a separable effect is called \textit{an entangled effect}.
We call observables that correspond to separable effect valued measures \textit{separable observables} and observables that are not separable \textit{entangled}.
  In the case of a finite outcome, these definitions are consistent with entangled measurements in Ref.~\cite{Vertesi2011}.

Let us consider a binary measurement that has two outcomes $+1$ and $-1$.
Let $\obM_1$, $\obM_{-1}\in \Eff{\complex^d}$ be POVM elements with output $+1$ and $-1$ respectively.
Let us write the Hermitian operator representing the expectation operator as $\obM = \obM_1 - \obM_{-1}= 2 \obM_1 - \id$.
We can safely identify this operator M with a binary observable $\Set{\obM_1, \obM_{-1}}$.
Let $\alpha_d$ be a constant $d/(d-1)$.
We define \textit{the difference from ignorance} by $E(\rho, \obM)\coloneqq \alpha_d \tr{(\rho - \id/d)\obM}/2 = \alpha_d\tr{(\rho-\id/d)\obM_1}$.
\begin{lem}
  \label{lem:ineq-dual}
  Let $X$ be an operator satisfying $0 \leq X \leq \id$ and $\rho$ be a state on
  a finite-dimensional Hilbert space $\complex^d$.
  The following inequality holds:
  \begin{equation*}
    -1 \leq \alpha_d\tr{\left(\rho-\frac{\id}{d}\right) X}\leq 1.
  \end{equation*}
\end{lem}
\begin{proof}
  Since $\rho$ is self-adjoint,
  $\rho$ is diagonalizable to $\rho = U \rho_D U^\dagger$ such that $\rho_D$ is a diagonal matrix and $U$ is a unitary matrix.
  Therefore,
  \begin{align}
    & \left| \alpha_d\tr{\left(\rho-\frac{\id}{d}\right) X} \right|\nonumber\\
    &= \left| \alpha_d\tr{
      \begin{pmatrix}
        \rho_{11}-\frac{1}{d} & & \\
        & \ddots & \\
        & & \rho_{dd}-\frac{1}{d}
      \end{pmatrix}
      X'} \right|,
  \end{align}
  where $X'=U^\dagger X U$.
  Since $X'$ is also an effect,
  \begin{align}
    &= \left| \alpha_d\tr{
      \begin{pmatrix}
        \rho_{11}-\frac{1}{d} & & \\
        & \ddots & \\
        & & \rho_{dd}-\frac{1}{d}
      \end{pmatrix}
      X'} \right|,\nonumber\\ 
    &\leq \alpha_d\sum_{\substack{i\in \set{1,\dots,d} \\ \rho_{ii}-1/d>0}} \left( \rho_{ii}-\frac{1}{d}\right) \nonumber\\
    &\leq \alpha_d\left(1-\frac{1}{d}\right)=1
  \end{align}
  holds.
  The last inequality follows from $\sum_i \rho_{ii}=1$ and the fact that the LHS is greater if the number of the term $-1/d$ is smaller.
  The last equality holds when state $\rho$ is a pure state.
\end{proof}

The difference from ignorance for a pair of quantum systems is defined as
\begin{equation*}
  E(\rho^A, \rho^B, \obM) = \frac{\alpha_{d_A}\alpha_{d_B}}{2} \tr{\left(\rho^A-\frac{\id}{d_A}\right) \otimes \left(\rho^B-\frac{\id}{d_B}\right)\obM}.
\end{equation*}
Now we define $ D = E(\rho^A_0, \rho^B_0, \obM) + E(\rho^A_0, \rho^B_1, \obM) + E(\rho^A_1, \rho^B_0, \obM) - E(\rho^A_1, \rho^B_1, \obM)$.
Hereinafter, we consider a limited class of effects.
We assume that an observable $\obM$ satisfies either $\tr{\obM_1}\leq 1$ or $\tr{\obM_{-1}}\leq 1$.
The condition $\tr{\obM_1}\leq 1$ is satisfied if and only if $\obM_1$ is written as
$\obM_1 = \sum \lambda_i E_i$ with a family of rank $1$ effects $\Set{E_i}$ and $0\leq \lambda_i \leq 1$ satisfying $\sum_i \lambda_{i} \leq 1$.
This assumption can be removed by a postprocessing (or coarse-graining).
The postprocessing procedure is as follows:
Without loss of generality, one can assume $\tr{\obM_1}>1$. 
If one obtains the outcome $+1$, one leaves it with probability $p=\frac{1}{\tr{\obM_1}}$ and flips the outcome to $-1$ with probability $1-p$.
If one obtains the outcome $-1$, one leaves it.
The postprocessed effect is $\obM'_1 = \frac{\obM_1}{\tr{\obM_1}}$, which satisfies the condition.
Since the set of separable effects is convex, this postprocessing preserves separability of effects.
Thus, one can renormalize the effect so that it satisfies the assumption without breaking the separability of effects.

The following theorem is our main result:
\begin{thm}[Dual version of Bell-CHSH inequality]
  \label{thm:gen-dual-bell-chsh}
  We assume that a bipartite binary observable $\obM$ with two outcomes $+1$ and $-1$ satisfies either $\tr{\obM_1}\leq 1$ or $\tr{\obM_{-1}}\leq 1$.
	If this observable $\obM$ is separable, then \textit{the dual version of the Bell--CHSH inequality}
  \begin{equation}
    \label{eq:dual-bell}
    -2 \leq D \leq 2
  \end{equation}
	is satisfied.
\end{thm}
\begin{proof}
  We assume $\tr{\obM_1}\leq 1$ (in the other case, $\tr{\obM_{-1}}\leq 1$ similarly applies.).
  Because $\obM$ is separable, an effect $\obM_1=\sum_i \lambda_i P_A^i\otimes P_B^i$ can be decomposed by using $\lambda_i \geq 0$.
  We can choose $P_A^i$ and $P_B^i$ as rank 1 projections.
  Thus
  \begin{align}
    1 &\geq \tr{\obM_1} \nonumber \\
      &= \tr{\sum_i \lambda_i P_A^i\otimes P_B^i} \nonumber \\
      &= \sum_i \lambda_i
  \end{align}
  holds.
  Therefore, $\obM_1$ can be decomposed as $\obM_1=\sum_i \lambda_i P_A^i\otimes P_B^i $ with $\lambda_i\geq 0$ and $\sum_i \lambda_i \leq 1$.

  For a product effect $P_A\otimes P_B$, it holds that
  \begin{align}
    E &= \frac{\alpha_{d_A}\alpha_{d_B}}{2}\nonumber\\&\quad\tr{\left(\rho_i^A-\frac{\id}{d_A}\right)\otimes \left(\rho_j^B-\frac{\id}{d_B}\right)(2P_A\otimes P_B-\id)} \nonumber \\
      &= \alpha_{d_A}\alpha_{d_B}\tr{\left(\rho_i^A-\frac{\id}{d_A}\right)P_A}\tr{ \left(\rho_j^B-\frac{\id}{d_B}\right)P_B} \nonumber \\
      &= A_i B_j,
  \end{align}
  where $A_i = \alpha_{d_A}\tr{(\rho_i^A-\id/d_A)P_A}$ and $B_j = \alpha_{d_B}\tr{ (\rho_j^B-\id/d_B)P_B}$.
  From Lemma \ref{lem:ineq-dual}, one can see that $|A_i|\leq 1$ and $|B_i|\leq 1$ hold.
  Thus we have
  \begin{align}
    |D|
    \leq& |A_0B_0+A_0B_1+A_1B_0-A_1B_1| \nonumber \\
    \leq& |A_0||B_0+B_1|+|A_1||B_0-B_1| \nonumber \\
    \leq& |B_0+B_1|+|B_0-B_1| \nonumber \\
    \leq& 2.
  \end{align}
  This can be regarded as the dual version of the Bell--CHSH inequality for a product effect.
  For a convex combination of a product effect, we have a similar inequality because of the linearity.
\end{proof}
This theorem holds for the separable effect $\obM_1\in E(\Hil_A)\otimes_{\min}E(\Hil_B)$.
If an entangled effect exists satisfying the conditions stated in Theorem~\ref{thm:gen-dual-bell-chsh},
the effect may violate this inequality.

\section{Experiment}\label{sec:experiment}
To show that the violation of our inequality can be demonstrated experimentally,
we used a cloud quantum computer, the IBM Quantum Experience, which is open to the public.
Because the quantum computer uses a qubit ($\Hil = \complex^2$) as its unit,
we executed our experiment for the two-qubit case.
Figure~\ref{fig:quantum-circuit} shows the situation for our experiment.
\begin{figure}
  \centering
  \includegraphics[width=0.6\linewidth]{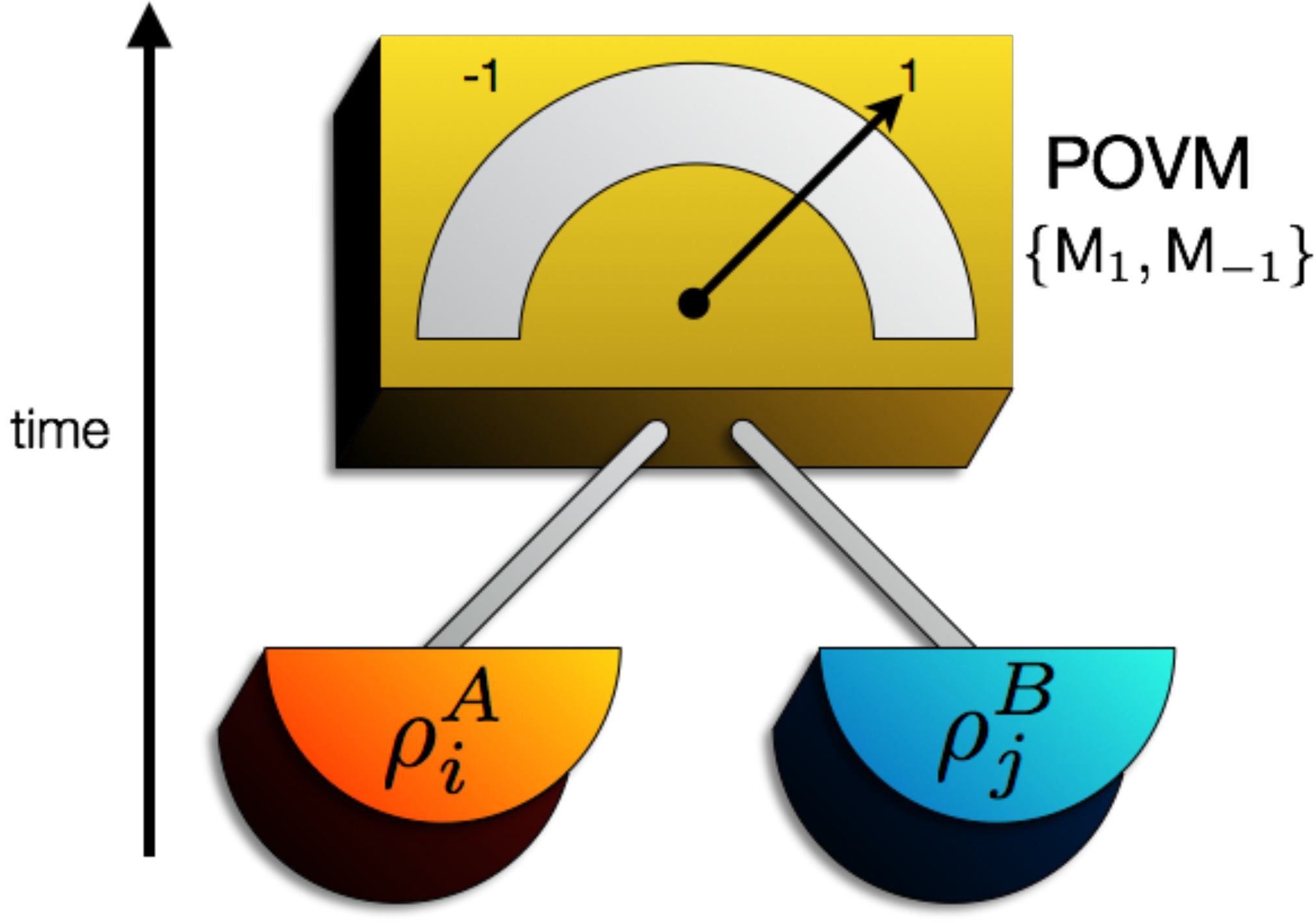}
  \caption{
    Our experiment consists of preparation of two states $\rho^A_i$ and $\rho^B_j$
    and one nonlocal observable $\obM$.
    The preparation of mixed states is provided by a probabilistic mixture of the preparation of a pure state.
    We perform a Bell measurement as a measurement of the nonlocal observable
    in the experiment.
  }
\label{fig:quantum-circuit}
\end{figure}
We perform a Bell measurement for the complete mixed state $\frac{\id_4}{4}$ (Figure~\ref{fig:compmixed-bellmeas}).
\begin{figure}[t]
  \centering
  \includegraphics[width=0.8\linewidth]{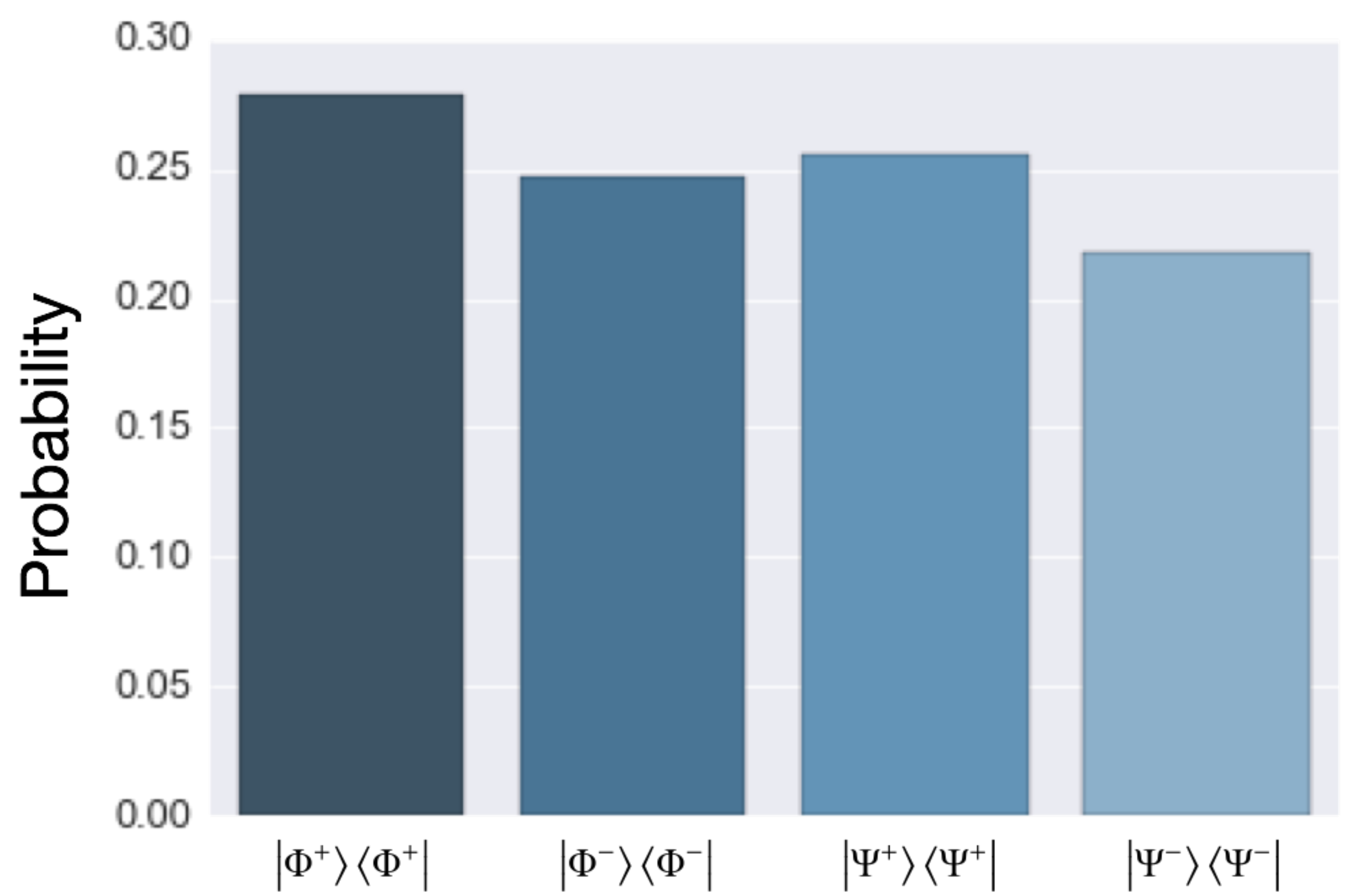}
  \caption{%
    This chart shows the probability distribution when a Bell measurement is made on the completely mixed state.
    Each probability is
$\tr{\frac{\id_4}{4}\Ket{\Phi^+}\Bra{\Phi^+}}= 0.2792$,
$\tr{\frac{\id_4}{4}\Ket{\Phi^-}\Bra{\Phi^-}}= 0.2474$,
$\tr{\frac{\id_4}{4}\Ket{\Psi^+}\Bra{\Psi^+}}= 0.2552$,
$\tr{\frac{\id_4}{4}\Ket{\Psi^-}\Bra{\Psi^-}}= 0.2182$.}
  \label{fig:compmixed-bellmeas}
\end{figure}
Because our theorem requires the condition $\tr{\frac{\id}{4} M_x} \leq \frac{1}{4}$,
we use the value $D$ for $\obM_1=\Ket{\Phi^-}\Bra{\Phi^-}$.
We prepare the states as
$\rho_0^A=\ket{0}\bra{0}$,
$\rho_1^A=\ket{+}\bra{+}$,
$\rho_0^B=\frac{2+\sqrt{2}}{4}\ket{0}\bra{0}-\frac{\sqrt{2}}{4}\ket{0}\bra{1}-\frac{\sqrt{2}}{4}\ket{1}\bra{0}+\frac{2-\sqrt{2}}{4}\ket{1}\bra{1}$, and
$\rho_1^B=\frac{2+\sqrt{2}}{4}\ket{0}\bra{0}+\frac{\sqrt{2}}{4}\ket{0}\bra{1}+\frac{\sqrt{2}}{4}\ket{1}\bra{0}+\frac{2-\sqrt{2}}{4}\ket{1}\bra{1}$.
In this setting, we obtain
\begin{equation}
D = 2.573 \pm 0.035.
\end{equation}
which shows the violation of the dual Bell--CHSH inequality.
Therefore the effect of a Bell measurement is confirmed to be an entangled effect.
This result rules out the maximal tensor product state space.

In the above experimental setting, the value $D$ theoretically takes $2\sqrt{2}$.
Our experiment does not reach this value because the quantum computer is noisy and has an error.
Moreover, this value of $2\sqrt{2}$ is an upper bound for quantum theory and can be regarded as a dual version of the Tsirelson bound~\cite{Cirelson1980}.
This follows from evaluating the square root of an operator norm of ${\left(\sum_{ij}{(-1)}^{ij}\alpha_{d_A}\alpha_{d_B}\left(\rho_i^A-\frac{\id}{d_A}\right)\otimes\left(\rho_j^B-\frac{\id}{d_B}\right)\right)}^2$.

\section{Dual Bell--CHSH inequality and quantum teleportation}\label{sec:qtelepo}
We now discuss the relation between the dual version of the Bell--CHSH inequality and quantum teleportation~\cite{Bennett1993} in qubit systems.
Quantum teleportation is a protocol for transporting an unknown state between two parties (called Alice and Bob) by using a pre-shared entangled state and local operations and classical communication (LOCC).
While not all entangled states are useful for quantum teleportation,
all qubit states that violate the Bell--CHSH inequality are useful for quantum teleportation~\cite{Horodecki1996a}.
Here we consider the case where the entanglement of two parties is maximal but Alice's Bell measurement is incomplete.
Let $\Set{\obA_i}_{i=0,1,2,3}$ be a POVM of Alice's entangled measurement, where each effect is represented as $\obA_i=\frac{1}{4}(\id\otimes\id +
\vec{r}_i\cdot\vec{\sigma}\otimes\id + \vec{s}_i\cdot\vec{\sigma}+\sum_{n,m=1}^{3}t_i^{nm}\sigma_n\otimes\sigma_m)$.
Alice prepares an unknown pure state $\phi$ and then performs a measurement of $\obA$ for her unknown state and one of the entangled state.
The probability that Alice obtains an outcome $i$ is $p_i=\tr{(\ket{\phi}\bra{\phi}\otimes \ket{\Phi^+}\bra{\Phi^+}) (\obA_i\otimes \id)}$.
When $i$ occurs, Bob performs a feedback unitary operation $U_i$.
The output state $\rho_i$ is $\frac{1}{p_i}\ptr{A}{(\id\otimes U_i)(\ket{\phi}\bra{\phi}\otimes \ket{\Phi^+}\bra{\Phi^+}) (\obA_i\otimes \id)(\id\otimes U_i^*)}$.
We define the average fidelity as
\begin{equation}
  \mathcal{F} = \int_S \mathrm{d}M(\phi) \sum_{i=0}^{3} p_i \Braket{\phi|\rho_i\phi},
\end{equation}
where $M$ is a uniform measure on the Bloch sphere.
From the duality and the calculation of Ref.~\cite{Horodecki1996a}, one can show that the maximum value of the average fidelity $\mathcal{F}_{\max}$ is
\begin{equation}
  \mathcal{F}_{\max}\coloneqq \max_{\set{U_i}_i} \mathcal{F}=\frac{1}{2}\left( 1 + \frac{1}{12}\sum_{i=0}^{3}\tr{\sqrt{T_i^*T_i}} \right),
\end{equation}
where $T_i$ is matrix whose matrix elements are $t_i^{nm}=\tr{\sigma_n\otimes\sigma_m\obA_i}$.
If $\mathcal{F}_{\max}$ is greater than $\frac{2}{3}$ then the observable $\obA$ is
useful for quantum teleportation.
Therefore, to be useful for quantum teleportation, an observable $\obA$ needs to satisfy
$\sum_i \tr{\sqrt{T_i^*T_i}} > 4$.
The following corollary is convenient to discuss quantum teleportation and entanglement of an observable.
This corollary seems to be the dual version of Theorem~1 in Ref.~\cite{Horodecki1995}
\begin{cor}
  For the Hilbert space $\complex^2\otimes\complex^2$,
  let an effect $\obM_1$ be $\frac{1}{4}(\id\otimes\id +
  \vec{r}\cdot\vec{\sigma}\otimes\id + \vec{s}\cdot\vec{\sigma}+\sum_{n,m=1}^{3}t^{nm}\sigma_n\otimes\sigma_m)$.
  Let $T$ be a matrix whose matrix elements are $t^{nm}=\tr{\sigma_n\otimes\sigma_m\obM_1}$.
  Inequality~\eqref{eq:dual-bell} holds for the effect $\obM_1$ and some states $\rho_i^A$ and $\rho_j^B$
  if and only if $\lambda_1(T^*T)+\lambda_2(T^*T)\leq1$
  where $\lambda_k(A)$ is the $k$th largest eigenvalue of $A$.
\end{cor}
\begin{proof}
  We assume that $\rho_i^A=\frac{1}{2}(\vec{a}_i\cdot\vec{\sigma}+\id)$, $\rho_j^B=\frac{1}{2}(\vec{b}_j\cdot\vec{\sigma}+\id)$.
  By a simple calculation, one has $E(\rho_i^A,\rho_j^B,\obM)=\braket{\vec{a}_i|T \vec{b}_j}$.
  Therefore $D=\braket{\vec{a}_0|T(\vec{b}_0+\vec{b}_1)}+\braket{\vec{a}_1|T(\vec{b}_0-\vec{b}_1)}$.
  The maximum value of $D$ is $2\sqrt{\lambda_1(T^*T)+\lambda_2(T^*T)}$.
  $2\sqrt{\lambda_1(T^*T)+\lambda_2(T^*T)}\leq 2$ if and only if $\lambda_1(T^*T)+\lambda_2(T^*T)\leq 1$.
\end{proof}

If these entangled effects $\obA_i$ violate the dual version of the Bell--CHSH inequality,
one has $\tr{\sqrt{T_i^*T_i}} > 1$ from Corollary~1.
Thus, the entangled observable that violates the dual Bell--CHSH inequality is useful for quantum teleportation.

\section{Discussion and conclusion}\label{sec:conclusion}
We introduced an entangled effect and the dual version of the Bell--CHSH inequality, which is satisfied by separable effects in a certain class of effects.
Note that some entangled effects cannot be detected by our inequality.
For example, for sufficiently small $\epsilon$, $\epsilon (\Ket{00}+\Ket{11})(\Bra{00}+\Bra{11})$ is an entangled effect but satisfies our inequality.
To develop further the method to be applicable to a wider class of effects is an important future problem.
Conversely, the present result is strong enough to rule out the maximal state space.
We show that the maximal state space for a pair of quantum systems is unphysical
because the effect space dual to the maximal state space violates the dual version of the Bell--CHSH inequality.
However, our theory is not strong enough to single out the composite system.
We hope to study this important issue in future work.

An effect that can be realized by LOCC is a separable effect.
However, it remains unclear whether all separable effects can be realized by a LOCC operation.
LOCC operations preserve the separability of effects because LOCC operations are separable operations.
Giving a completely operational meaning to separable effects is an open problem.

In this Letter, we assume that the local state space is quantum.
It is natural to consider the case in which this assumption is removed and an extension is made to generalized probabilistic theories (GPTs)~\cite{Janotta2014}.
Our inequality requires the completely mixed state and the dimension of the Hilbert space.
For some special GPTs, the corresponding values can be defined.
For example, gbit (square-shaped state space) has the center of gravity, which corresponds to the completely mixed state.
The corresponding the dimension of the Hilbert space is two.
Therefore, our inequality can be extended to gbit.
The tensor product of gbit is not unique.
The most famous example is the PR box, which is the maximal tensor product state space of two gbits.
However, the PR box does not have the entangled effect, so the PR box does not violate our inequality.
Conversely, the minimal tensor product state space of two gbits has an entangled effect, which violates our inequality and has Tsirelson-like bound of four.
To extend to the case of more general GPTs, we need to find a new inequality that does not depend on the completely mixed state and the dimension of Hilbert space.

The high performance of quantum information processing is due to entangled states.
In addition, we show that for the first time that entangled observables whose effects violate the dual version of the Bell--CHSH inequality is useful for quantum teleportation.
We expect that entangled effects are needed to support the richness of state transformations
and enhance the capability of quantum information processing.

\section*{Acknowledgements}

The author thanks T. Miyadera, Y. Kuramochi, and M. Kobayashi for fruitful discussions.
We acknowledge use of the IBM Quantum Experience for this work.
The views expressed are those of the authors and do not reflect the official policy or position of IBM or the IBM Quantum Experience team.
The author acknowledges support by Grant-in-Aid for JSPS Research Fellow (JP18J10310).

\bibliographystyle{elsarticle-num}
\bibliography{main.bib}

\end{document}